\documentclass[a4paper,UKenglish,cleveref, autoref]{lipics-v2019}

\usepackage{amsmath,amssymb}
\usepackage{amsthm}
\usepackage{xcolor}
\usepackage{paralist}

\newcommand{\eset}[1]{\{\,#1\,\}}
\newcommand{\lreset}[1]{\left\{\,#1\,\right\}}
\newcommand{\seq}[1]{\langle#1\rangle}

\newcommand{\Nats}{\mathbb{N}}

\newcommand{\A}{\mathcal{A}}
\newcommand{\Bb}{\mathcal{B}}
\newcommand{\Gg}{\mathcal{G}}
\newcommand{\Ss}{\mathcal{S}}
\newcommand{\Ww}{\mathcal{W}}
\newcommand{\Uu}{\mathcal{U}}

\title{Alternating Weak Automata from Universal Trees} 

\author{Laure Daviaud}{City, University of London, UK}{Laure.Daviaud@city.ac.uk}{}{}

\author{Marcin Jurdzi\'nski}{University of Warwick, UK}{Marcin.Jurdzinski@warwick.ac.uk}{}{}

\author{Karoliina Lehtinen}{University of Liverpool, UK}{k.lehtinen@liverpool.ac.uk}{}{}

\authorrunning{L. Daviaud, M. Jurdzi\'nski and K. Lehtinen}

\Copyright{Laure Daviaud, Marcin Jurdzi\'nski and Karoliina Lehtinen}



\relatedversion{}

\supplement{}

\funding{%
This work has been supported by the EPSRC grants EP/P020992/1 and
EP/P020909/1 (Solving Parity Games in Theory and Practice). %
}

\acknowledgements{%
We thank Moshe Vardi for encouraging us to bring the state-space
blow-up of alternating parity to alternating weak automata translation in
line with the state-of-the-art complexity of solving parity games.}

\nolinenumbers 

\hideLIPIcs  

\EventEditors{Wan Fokkink and Rob van Glabbeek}
\EventNoEds{2}
\EventLongTitle{30th International Conference on Concurrency Theory (CONCUR 2019)}
\EventShortTitle{CONCUR 2019}
\EventAcronym{CONCUR}
\EventYear{2019}
\EventDate{August 27--30, 2019}
\EventLocation{Amsterdam, the Netherlands}
\EventLogo{}
\SeriesVolume{140}
\ArticleNo{14}

\begin{document}

\maketitle

\begin{abstract}
  An improved translation from alternating parity automata
  on infinite words to alternating weak automata is given.  
  The blow-up of the number of states is related to the size of the 
  smallest universal ordered trees and hence it is quasi-polynomial,
  and it is polynomial if the asymptotic number of priorities is at
  most logarithmic in the number of states.
  This is an exponential improvement on the translation of Kupferman
  and Vardi (2001) and a quasi-polynomial improvement on the
  translation of Boker and Lehtinen (2018). 
  Any slightly better such translation would
  (if---like all presently known such translations---it is efficiently
  constructive) 
  lead to algorithms for solving parity games that are asymptotically
  faster in the worst case than the current state of the art
  (Calude, Jain, Khoussainov, Li, and Stephan, 2017;
  Jurdzi\'nski and Lazi\'c, 2017;
  and Fearnley, Jain, Schewe, Stephan, and Wojtczak, 2017),
  and hence it would yield a significant breakthrough. 
\end{abstract}

\section{Introduction}

The influential class of regular languages of infinite words 
(often called the $\omega$-regular languages)
is defined to consist of all the languages of infinite words that are
recognized by finite non-deterministic B\"uchi automata.
The theory of $\omega$-regular languages is quite well understood.  
In particular, it is known that deterministic B\"uchi automata are not 
sufficiently expressive to recognize all the $\omega$-regular
languages, but deterministic automata with the so-called parity
acceptance conditions are, and that the class of $\omega$-regular
languages is closed under complementation.
Effective constructions for determinization and complementation of
B\"uchi automata are important tools both in theory and in
applications, and both are known to require exponential blow-ups of
the numbers of states in the worst case.

For applications in logic, it is natural to enrich automata models by
the ability to alternate between non-deterministic and universal
transitions~\cite{MSS86,KVW00}.
It turns out that alternating parity automata are no more expressive
than non-deterministic B\"uchi automata, and hence neither allowing
alternation, nor the richer parity acceptance conditions, increase
expressiveness;
this testifies to the robustness of the class of $\omega$-regular 
languages. 
On the other hand, alternation increases the expressive power of
automata with the so-called weak acceptance conditions:
non-deterministic weak automata are not expressive enough to recognize
all $\omega$-regular languages, but alternating weak automata are.
The weak acceptance conditions are significant due to their
applications in logic~\cite{MSS86} and thanks to their favourable
algorithmic properties~\cite{KVW00}. 

Given that alternating weak automata are expressive enough to
recognize all the $\omega$-regular languages, a natural question is
whether, and to what degree, alternating weak automata are less
succinct than alternating B\"uchi or alternating parity automata. 
Another way of stating this question is what blow-up in the number of
states is sufficient or required for translations from alternating
parity or alternating B\"uchi automata to alternating weak automata. 
The first upper bound for the blow-up of a translation from
alternating B\"uchi to alternating weak automata was doubly
exponential, obtained by combining a doubly-exponential
determinization construction~\cite{DH94} and a linear translation from 
deterministic parity automata to weak alternating
automata~\cite{MSS86,Lin88}.  
This has been improved considerably by Kupferman and Vardi who have
given a quadratic translation from alternating B\"uchi to alternating
weak automata~\cite{KV01}, and then they have generalized it to a
translation from alternating parity automata with $n$ states and $d$
priorities to alternating weak automata, whose blow-up is
$n^{d+O(1)}$, i.e., exponential in the number of priorities
in the parity acceptance condition~\cite{KV98}. 

Understanding the exact trade-off between the complexity of the
acceptance condition---weak, B\"uchi, or parity, the latter measured
by the number of priorities---and the number of states in an automaton
is interesting from the algorithmic point of view. 
For example, the algorithmic problems of checking emptiness of
non-deterministic parity automata on infinite trees, of model checking
for the modal $\mu$-calculus, of solving two-player parity games, and 
of checking emptiness of alternating parity automata on infinite
words over a one-letter alphabet, are all polynomial-time equivalent.  
Since checking emptiness of alternating weak automata on words over a
one-letter alphabet can be done in linear time, it follows that a
translation from alternating parity automata to alternating weak
automata implies an algorithm for solving parity games whose
complexity matches the blow-up of the number of states in the
translation. 

The first quasi-polynomial translation from alternating parity
automata to alternating weak automata was given recently by Boker and
Lehtinen~\cite{BL18}. 
They have used the register technique, developed by
Lehtinen~\cite{Leh18} for parity games,  
to provide a translation from alternating parity
automata with $n$ states and $d$ priorities to alternating parity
automata with $n^{\Theta(\log(d/\log n))}$ states and $\Theta(\log n)$
priorities; 
combined with the exponential translation of Kupferman and
Vardi~\cite{KV98}, this yields an alternating parity to alternating
weak translation whose blow-up of the number of states  
is~$n^{\Theta(\log n \cdot \log(d/\log n))}$. 

The main result reported in this paper is that another
technique---universal trees~\cite{JL17,CDFJLP19}, also developed to  
elucidate the recent major advance in the complexity of solving parity
games due to Calude, Jain, Khoussainov, Li, and
Stephan~\cite{CJKLS17}---can be used to further reduce the
state-space blow-up in the translation from alternating parity
automata to alternating weak automata.   
We give a translation from alternating parity automata with $n$ states
and $d$ priorities to alternating B\"uchi automata, whose state-space
blow-up is proportional to the size of the 
smallest $(n, d/2)$-universal trees~\cite{CDFJLP19}, which is
polynomial in~$n$ if $d = O(\log n)$ and it is $n^{\lg(d/\lg n)+O(1)}$ if 
$d = \omega(\log n)$. 
When combined with Kupferman and Vardi's quadratic translation of
alternating B\"uchi to alternating weak automata~\cite{KV01}, we get
the composite blow-up of the form $n^{O(\log(d/\log n))}$, down from
Boker and Lehtinen's blow-up 
of~$\left(n^{\Theta(\log(d/\log n))}\right)^{\log n}$.  

%

The necessary size of the state-space blow-up when going from
alternating parity automata to alternating weak automata is wide open: 
the best known lower bound is $\Omega(n \log n)$~\cite{KV01}, closely
related to the $2^{\Omega(n \log n)}$ lower bound on B\"uchi
complementation~\cite{Mic88}, while the best upper bounds are
quasi-polynomial. 
On the other hand, the blow-up for the parity to weak translation that
we obtain nearly matches the current state-of-the-art quasi-polynomial
upper bounds on the complexity of solving parity
games~\cite{CJKLS17,JL17,FJSSW17}.  
It follows that any significant improvement over our translation would
lead to a breakthrough improvement in the complexity of solving parity
games.

The exponential translation from alternating parity automata to
alternating weak automata due to Kupferman and Vardi~\cite{KV98} is
done by a rather involved induction on the number of priorities. 
For an automaton with $d$ priorities, it goes through a
sequence of $d$ intermediate automata of a generalized type, which
they call parity-weak alternating automata.
In contrast, our construction is significantly more streamlined and
transparent; in particular, it avoids introducing a new class of
hybrid parity-weak automata. 
We first establish a hierarchical decomposition of runs
of alternating parity automata as a generalization of the
decomposition of runs of alternating co-B\"uchi automata due to
Kupferman and Vardi~\cite{KV01}, and then we use the recently
introduced universal trees~\cite{JL17,CDFJLP19} to construct an
alternating B\"uchi automaton, which is a parity automaton with
just $2$ priorities. 
Our work is yet another application of the recently introduced notion
of universal trees~\cite{JL17,CDFJLP19}. 
Such applications typically focus on algorithms for solving
games~\cite{JL17,DJL18,CDFJLP19,CF19};
our work is the first whose primary focus is on automata.

In addition to universal trees, we use a notion of lazy progress
measure.
Unlike the standard parity progress measures, which can be recognised
by safety automata but require an explicit bound to be known on the
number of successive occurrences of odd priorities, lazy progress
measures are recognised by B\"uchi automata and can deal with finite
but unbounded numbers of occurrences of successive odd priorities.   
This  B\"uchi automaton is similar to (but more concise than) the
automaton used to characterise parity tree automata that recognise
co-B\"uchi recognisable tree languages~\cite{LQ17}, itself a
generalisation of automata used to decide the weak definability of
tree languages given as B\"uchi automata~\cite{SW16,CKLV13}.  

A similar concept to our lazy progress measures was already introduced
by Klarlund for complementation of B\"uchi and Streett automata on
words~\cite{Kla91}.
Klarlund indeed proves a result that is equivalent to one of our key
lemmas on parity progress measures. 
Our proof, however, is more constructive, and it explicitly
provides a hierarchical decomposition, which clearly describes the
structure of accepting run dags of parity word automata. 
Moreover---unlike Klarlund's proof, which relies on the result about 
Rabin measures~\cite{KK95}---our proof is self-contained.
We suspect that the opaqueness of Klarlund's paper~\cite{Kla91} may
have been responsible for attracting less attention and shallower 
absorption by the research community than it deserves.
In particular, some of the techniques and results that he presents
there have been rediscovered and refined by various authors, often
much later~\cite{KV01,KV98,Jur00,JL17,CDFJLP19}, including this 
work.
We hope that our paper will help a wider and more thorough reception
and appreciation of Klarlund's work.

\section{Alternating automata}

For a finite set~$X$, we write $\Bb^+(X)$ for the set of positive
Boolean formulas over~$X$.
We say that a set~$Y \subseteq X$ \emph{satisfies} a formula
$\varphi \in \Bb^+(X)$ if~$\varphi$ evaluates to $\mathtt{true}$ when
all variables in~$Y$ are set to $\mathtt{true}$ and all variables in
$X \setminus Y$ are set to $\mathtt{false}$. 
For example, the sets $\eset{x, y}$ and $\eset{x, z}$ satisfy the
positive Boolean formula $x \wedge (y \vee z)$, but the set
$\eset{y, z}$ does not. 
An \emph{alternating automaton} has a finite set~$Q$ of \emph{states},
an \emph{initial state} $q_0 \in Q$, a finite \emph{alphabet}
$\Sigma$, and a transition function
$\delta : Q \times \Sigma \to \Bb^+(Q)$. 
Alternating automata allow to combine both \emph{non-deterministic}
and \emph{universal} transitions;
disjunctions in transition formulas model the non-deterministic
choices and conjunctions model the universal choices.

We consider alternating automata as acceptors of infinite words.
Whether infinite sequences of states in runs of such automata are
\emph{accepting} or not is determined by an
\emph{acceptance condition}.
Here, we consider \emph{parity}, \emph{B\"uchi}, \emph{co-B\"uchi},
\emph{weak}, and \emph{safety} acceptance conditions.
In a \emph{parity} condition, given by a
\emph{state priority function} $\pi : Q \to \eset{0, 1, 2, \dots, d}$  
for some positive even integer~$d$, an infinite sequence of states is   
accepting if the largest state priority that occurs infinitely many
times is even.
\emph{B\"uchi} conditions are a special case of parity conditions
in which all states have priorities~$1$ or~$2$, and \emph{co-B\"uchi} 
conditions are parity conditions in which all states have
priorities~$0$ or~$1$.  

Let the \emph{transition graph} of an alternating automaton have an
edge $(q, r) \in Q \times Q$ if $r$ occurs in $\delta(q, a)$ for some
letter $a \in \Sigma$.  
We say that a parity automaton has a \emph{weak} acceptance condition 
if it is \emph{stratified}:
in every cycle in the transition graph, all states have the same
priority. 
Weak conditions are a special case of both B\"uchi and co-B\"uchi
conditions in the following sense: 
if the transition graph of a parity automaton is stratified, then
every infinite path in the transition graph satisfies each of the
following three conditions if and only if it satisfies the other two: 
\begin{compactitem}
\item
  the parity condition
  $\pi : Q \to \eset{0, 1, 2, \dots, d}$; 
\item
  the co-B\"uchi condition $\pi' : Q \to \eset{0, 1}$;
\item
  the B\"uchi condition $\pi'' : Q \to \eset{1, 2}$;
\end{compactitem}
where
$\pi'(q) = \pi(q) \bmod 2$ and $\pi''(q) = 2 - \pi'(q)$ for all
$q \in Q$. 

We say that a state is \emph{absorbing} if its only successor in the
transition graph is itself.
A parity automaton has a \emph{safety} acceptance condition if all of
its states have priority~$0$, except for the additional absorbing
$\mathtt{reject}$ state that has priority~$1$.
An automaton with a safety acceptance condition is stratified, and
hence safety conditions are a special case of weak conditions.

Whether an infinite word $w = w_0 w_1 w_2 \cdots \in \Sigma^\omega$ is
\emph{accepted} or \emph{rejected} by an alternating automaton~$\A$ is
determined by the winner of the following \emph{acceptance game}
$\Gg(\A, w)$. 
The set of positions in the game is the set $Q \times \Nats$ 
and the two players, Alice and Elvis, play in the following way.
The \emph{initial position} is~$(q_0, 0)$;
for every \emph{current position} $(q_i, i)$, first Elvis chooses a 
subset $P$ of~$Q$ that satisfies $\delta(q_i, w_i)$, then Alice picks
a state~$q_{i+1} \in P$, and $(q_{i+1}, i+1)$ becomes the next current  
position.
Note that Elvis can be thought of making the non-deterministic choices 
and Alice can be thought of making the universal choices in the
transition function of the alternating automaton. 
This interaction of Alice and Elvis yields an infinite sequence of
states $q_0, q_1, q_2, \dots$, and whether Elvis is declared the
winner or not is determined by whether the sequence is
\emph{accepting} according to the acceptance condition of the
automaton.  
Acceptance games for parity, B\"uchi, co-B\"uchi, and weak conditions
are parity games, which are \emph{determined}~\cite{EJ91}:
in every acceptance game, either Alice or Elvis has a winning
strategy. 
We say that an infinite word $w \in \Sigma^\omega$ is \emph{accepted}
by an alternating automaton~$\A$ if Elvis has a winning strategy in
the acceptance game~$\Gg(\A, w)$, and otherwise it is
\emph{rejected}. 

A \emph{run dag} of an alternating automaton~$\A$ on an infinite
word~$w$ is a directed acyclic graph $G = (V, E, \rho : V \to Q)$,
where $V \subseteq Q \times \Nats$ is the set of vertices; 
successors (according to the directed edge relation~$E$) of every
vertex $(q, i)$ are of the form $(q', i+1)$; 
the following conditions hold:
\begin{compactitem}
\item
  $(q_0, 0) \in V$,
\item
  for every $(q, i) \in V$, the Boolean formula $\delta(q, w_i)$ is
  satisfied by the set of states~$p$, such that $(p, i+1)$ is a
  successor of $(q, i)$;  
\end{compactitem}
and $\rho$ projects vertices
onto the first component. 
Note that every vertex in a run dag has a successor, and hence every
maximal path is infinite.  
We say that a run dag of an automaton~$\A$ is \emph{accepting} if
the sequence of states on every infinite path in the run dag is 
accepting according to the accepting condition of~$\A$.  
The \emph{positional} determinacy theorem for parity
games~\cite{EJ91} implies that an infinite word~$w$ is
accepted by an alternating automaton~$\A$ with a parity
(or B\"uchi, co-B\"uchi, weak, or safety) condition if and only if
there is an accepting run dag of~$\A$ on~$w$.
In other words, run dags are compact representations of (positional)
winning strategies for Elvis in the acceptance games. 

Run dags considered here are a special case of \emph{layered dags},
whose vertices can be partitioned into sets $L_0, L_1, L_2, \dots$,
such that every edge goes from some layer~$L_i$ to the next
layer~$L_{i+1}$.
We define the \emph{width} of a layered dag with an infinite number of
layers~$L_0, L_1, L_2, \dots$ to be $\liminf_{i \to \infty} |L_i|$.
Note that the width of a run dag of an alternating automaton is
trivially upper-bounded by the number of states of the automaton.

\section{From co-B\"uchi and B\"uchi to weak}
\label{section:from-co-buchi-to-weak}

In this section we summarize the results of Kupferman and
Vardi~\cite{KV01} who have given translations from alternating
co-B\"uchi and B\"uchi automata to alternating weak automata with only
a quadratic blow-up in the state space.
We recall the decomposition of co-B\"uchi accepting run dags of 
Kupferman and Vardi in detail because it motivates and prepares the
reader for our generalization of their result to accepting parity run
dags.  
Our main technical result is a translation from alternating
parity automata to alternating B\"uchi automata with only a
quasi-polynomial blow-up in the state space, but the ultimate goal is
a quasi-polynomial translation from parity to weak automata.
Therefore, we also recall how Kupferman and Vardi use their quadratic
co-B\"uchi to weak translation in order to obtain a quadratic B\"uchi
to weak translation. 

\subsection{Co-B\"uchi progress measures}

The main technical concept that underlies Kupferman and
Vardi's~\cite{KV01} translation from alternating co-B\"uchi automata
to alternating weak automata is that of a \emph{ranking function} for
accepting run dags of alternating co-B\"uchi automata.
As Kupferman and Vardi themselves point out, ranking functions can be
seen as equivalent to Klarlund's 
\emph{progress measures}~\cite{Kla91}. 
We will adopt Klarlund's terminology because the theory of progress
measures for certifying parity conditions is very well
developed~\cite{EJ91,Kla91,KK95,Jur00,JL17,CDFJLP19} and our main goal
in this paper is to use a version of parity progress measures to give
a simplified, streamlined, and improved translation from alternating
parity to alternating weak automata. 

Let $G = (V, E, \pi : V \to \eset{0, 1})$ be a layered dag with vertex
priorities~$0$ or~$1$, and in which every vertex has a successor.
Note that all run dags of an alternating co-B\"uchi automaton are such
layered dags and if the automaton has~$n$ states then the width of the
run dag is at most~$n$.
(Observe, however, that while formally the third component
$\rho : V \to Q$ in a run dag maps vertices to states, here we instead
consider the labeling $\pi : V \to \eset{0, 1}$ that labels vertices
by the priorities of the states $\pi(v) = \pi(\rho(v))$.) 

A \emph{co-B\"uchi progress measure}~\cite{EJ91,KK95,Jur00} is a
mapping $\mu : V \to M$, where $(M, \leq)$ is a well-ordered set, such
that 
for every edge $(v, u) \in E$, we have
\begin{compactenum}
\item
  if $\pi(v) = 0$ then $\mu(v) \geq \mu(u)$, 
\item
  if $\pi(v) = 1$ then $\mu(v) > \mu(u)$. 
\end{compactenum}
It is elementary to argue that existence of a co-B\"uchi progress
measure on a graph is sufficient for every infinite path in the graph
satisfying the co-B\"uchi condition.  
Importantly, it is also necessary, which can be, for example, deduced
from the proof of positional determinacy for parity games due to
Emerson and Jutla~\cite{EJ91}.   
In other words, co-B\"uchi progress measures are witnesses for the
property that all infinite paths in a graph satisfy the co-B\"uchi
condition. 
The appeal of such witnesses stems from the property that while
certifying a global and infinitary condition, it suffices to verify
them locally by checking a simple inequality between the labels of the
source and the target of each edge in the graph. 

The disadvantage of progress measures as above is that on graphs of 
infinite size, such as run dags, the well-ordered sets of labels that
are needed to certify co-B\"uchi conditions may be of unbounded 
(and possibly infinite) size.
In order to overcome this disadvantage, and to enable
automata-theoretic uses of progress measure certificates, Klarlund has
proposed the following concept of lazy progress
measures~\cite{Kla91}. 
A \emph{lazy (co-B\"uchi) progress measure} is a mapping 
$\mu : V \to M$, where $(M, \leq)$ is a well-ordered set and 
$L \subset M$ is the set of \emph{lazy-progress} elements, and such
that:
\begin{compactenum}
\item
\label{item:non-inc}
  for every edge $(v, u) \in E$, we have $\mu(v) \geq \mu(u)$;
\item
\label{item:pri-1-lazy}
  if $\pi(v) = 1$ then $\mu(v) \in L$;
\item
\label{item:lazy-not-forever}
  on every infinite path in~$G$, there are infinitely many
  vertices~$v$ such that $\mu(v) \not\in L$. 
\end{compactenum}
It is elementary to prove the following proposition. 

\begin{proposition}
  If a graph 
  has a lazy co-B\"uchi progress measure
  then all infinite paths in it satisfy the co-B\"uchi condition.  
\end{proposition}
The following converse establishes the attractiveness of lazy
co-B\"uchi progress measures for certifying the co-B\"uchi conditions
on layered dags of bounded width, and hence for certifying accepting
run dags of alternating co-B\"uchi automata. 

\begin{lemma}[Klarlund~\cite{Kla91}]
  \label{lemma:small-coBuchi-lazy-pm}
  If all infinite paths in a layered dag
  $(V, E, \pi : V \to \eset{0, 1})$ satisfy the co-B\"uchi 
  condition and the width of the dag is at most~$n$, 
  then there is a lazy co-B\"uchi progress measure
  $\mu : V \to M$, where $M = \eset{1, 2, \dots,2n}$ and 
  $L = \eset{2, 4, 6, \dots, 2n}$. 
\end{lemma}

\begin{proof}
  We summarize a proof given by Kupferman and Vardi~\cite{KV01} that
  provides an explicit decomposition of the accepting run dag into
  (at most) $2n$ parts from which a lazy co-B\"uchi progress measure 
  can be straightforwadly defined.
  The proof by Klarlund~\cite{Kla91} is more succinct, but the former
  is more constructive and hence more transparent.

  Observe that if all infinite paths satisfy the co-B\"uchi condition 
  then there must be a vertex~$v$ whose all descendants 
  (i.e., vertices to which there is a---possibly empty---path 
  from~$v$)
  have priority~$0$; 
  call such vertices \emph{$1$-safe} in~$G_1 = G$. 
  Indeed, otherwise it would be easy to construct an infinite path
  with infinitely many occurrences of
  vertices of priority~$1$. 

  Let $S_1$ be the set of all the $1$-safe vertices in $G_1$,
  and let $G'_1 = G_1 \setminus S_1$ be the layered dag obtained
  from~$G_1$ by removing all vertices in~$S_1$.
  Note that there is an infinite path in the subgraph of~$G_1$
  induced by~$S_1$, and hence the width of~$G'_1$ is strictly smaller 
  than the width of~$G_1$.
  
  Let $R_1$ be the set of all vertices in~$G'_1$ that have only
  finitely many descendants;
  call such vertices \emph{transient} in~$G'_1$. 
  Let $G_2$ be the the layered dag obtained from~$G'_1$ by removing
  all vertices in~$R_1$.
  Since $G_2$ is a subgraph of~$G'_1$, the width of~$G_2$ is strictly
  smaller than the width of~$G_1$.
  Moreover, $G_2$ shares the key properties with~$G_1$:
  every vertex has a successor and hence all the maximal paths are
  infinite, 
  and all infinite paths satisfy the co-B\"uchi condition. 

  By applying the same procedure to $G_2$ that we have described
  for~$G_1$ above, we obtain the set $S_2$ of $1$-safe vertices
  in~$G_2$ and the set $R_2$ of vertices transient in~$G'_2$, and the
  layered dag~$G_3$---obtained from~$G_2$ by removing all vertices 
  in~$S_2 \cup R_2$---has the width that is strictly smaller than that 
  of~$G_2$.
  We can continue in this fashion until the graph~$G_{k+1}$, for some 
  $k \geq 1$, is empty. 
  Since the width of~$G$ is at most~$n$, and the widths of graphs
  $G_1, G_2, \dots, G_{k+1}$ are strictly decreasing, it follows that 
  $k \leq n$. 

  We define $\mu : V \to \eset{1, 2, \dots, 2n}$ by:
  \[
  \mu(v) =
  \begin{cases}
    2i-1 & \text{ if $v \in S_i$},
    \\
    2i & \text{ if $v \in R_i$},
  \end{cases}
  \]
  and note that it is routine to verify that if we let
  $L = \eset{2, 4, \dots, 2n}$ be the set of lazy-progress elements 
  then $\mu$ is a lazy co-B\"uchi progress measure. 
\end{proof}

\subsection{From co-B\"uchi and B\"uchi to weak}

In this section we present a proof of the following result. 

\begin{theorem}[Kupferman and Vardi~\cite{KV01}]
\label{thm:co-Buchi-quadratic}
  There is a translation that given an alternating co-B\"uchi
  automaton with $n$ states yields an equivalent alternating weak
  automaton with $O(n^2)$ states. 
\end{theorem} 



\begin{proof}
It suffices to argue that, given an alternating co-B\"uchi 
automaton~$\A = (Q, q_0, \Sigma, \delta, \pi : Q \to \eset{0, 1})$
with $n$ states, we can design an alternating weak automaton with
$O(n^2)$ states that guesses and certifies a dag run of~$\A$ together
with a lazy co-B\"uchi progress measure on it as described in 
Lemma~\ref{lemma:small-coBuchi-lazy-pm}. 
First we construct a \emph{safety} automaton $\Ss$ with $O(n^2)$
states that simulates the automaton~$\A$ while guessing a lazy
co-B\"uchi progress measure and verifying
conditions~\ref{item:non-inc}) and~\ref{item:pri-1-lazy}) of its
definition. 
Condition~\ref{item:lazy-not-forever}) will be later handled by
turning the safety automaton~$\Ss$ into a weak automaton~$\Ww$ by
appropriately assigning odd or even priorities to all states in~$\Ss$.  
We split the design of~$\Ww$ into those two steps so that we can
better motivate and explain the generalized constructions in 
Section~\ref{section:from-parity-to-buchi}. 

The safety automaton~$\Ss$ has the following set of states:

$$  Q \times \eset{2, 4, \dots, 2n}
  \,  
  \\
  \, \cup \,
  \left(\pi^{-1}(0) \times \eset{1, 3, \dots, 2n-1}\right)
  \, \cup \,
  \eset{\mathtt{reject}}\,;$$

its initial state is $(q_0, 2n)$; 
and its transition function $\delta'$ is obtained from the transition 
function~$\delta$ of~$\A$ in the following way:
for every state $(q, i)$, and for every $a \in \Sigma$, the formula
$\delta'\big((q, i), a\big)$ is obtained from $\delta(q, a)$ by
replacing every occurrence of state~$q' \in Q$ by the disjunction 
(i.e., a non-deterministic choice)  
\begin{equation}
\label{eq:disjunction}
  (q', i) \, \vee \, (q', i-1) \, \vee \, \cdots \, \vee \, (q', 1)
\end{equation}
where every occurrence $(q', j)$ for which $\pi(q')=1$ and $j$ is odd 
stands for the state $\mathtt{reject}$. 

In other words, the safety automaton~$\Ss$ can be thought of as
consisting of $2n$ copies $\A_{2n}, \A_{2n-1}, \dots, \A_1$ of~$\A$,
with the non-accepting states $\pi^{-1}(1)$ removed from the
odd-indexed copies $\A_{2n-1}, \A_{2n-3}, \dots, \A_1$, and in whose
acceptance games, Elvis always has the choice to stay in the current
copy of~$\A$ or to move to one of the lower-indexed copies of~$\A$.
Since the transitions of the safety automaton~$\Ss$ always respect the
transitions of the original co-B\"uchi automaton~$\A$, an accepting
run dag of~$\Ss$ yields a run dag of~$\A$
(obtained from the first components of the states $(q, i)$)
and a labelling of its vertices by numbers
in~$\eset{1, 2, \dots, 2n}$
(obtained from the second components of the states $(q, i)$). 
It is routine to verify that the design of the state set and of the
transition function of the safety automaton~$\Ss$ guarantees that the
latter labelling satisfies conditions~\ref{item:non-inc})
and~\ref{item:pri-1-lazy}) of the definition of a lazy co-B\"uchi 
progress measure, where the set of lazy-progress elements is
$\eset{2, 4, \dots, 2n}$.

By setting the state priority function
$\pi' : (q, i) \mapsto i+1$ for all non-$\mathtt{reject}$
states in~$\Ss$, and $\pi' : \mathtt{reject} \mapsto 1$, we obtain
from~$\Ss$ an automaton~$\Ww$ whose acceptance condition is weak
because---by design---the transition function is non-increasing
w.r.t.\ the state priority function. 
One can easily verify that the
addition of this weak acceptance condition to~$\Ss$ allows the
resulting automaton~$\Ww$, for every input word, to guess and verify a
lazy progress measure---if one exists---on a run dag of automaton~$\A$
on the input word, while~$\Ww$ rejects the input word otherwise. 
This completes our summary of the proof of
Theorem~\ref{thm:co-Buchi-quadratic}. 
\end{proof}

\begin{corollary}[Kupferman and Vardi~\cite{KV01}]
\label{thm:Buchi-quadratic}
  There is a translation that given an alternating B\"uchi automaton
  with $n$ states yields an equivalent alternating weak automaton with
  $O(n^2)$ states. 
\end{corollary}

The argument of Kupferman and Vardi is simple and it exploits the ease
with which alternating automata can be complemented.
Given an alternating B\"uchi automaton~$\A$ with $n$ states, first
complement it with no state space blow-up, obtaining an alternating
co-B\"uchi automaton with $n$ states, next use the translation from
Theorem~\ref{thm:co-Buchi-quadratic} to obtain an equivalent
alternating weak automaton with $O(n^2)$ states, and finally
complement the latter again with no state space blow-up, hence
obtaining an alternating weak automaton that is equivalent to~$\A$ and
that has $O(n^2)$ states.

\section{Lazy parity progress measures}

Before we introduce \emph{lazy parity progress measures}, we recall the 
definition of (standard) parity progress measures~\cite{JL17,CDFJLP19}.  
We define a \emph{well-ordered tree} to be a finite prefix-closed set of 
sequences of elements of a well-ordered set.
We call such sequences \emph{nodes} of the tree, and their components
are \emph{branching directions}.
We use the standard ancestor-descendant terminology to describe
relative positions of nodes in a tree. 
For example, $\seq{}$ is the \emph{root};
node $\seq{x, y}$ is the \emph{child} of the node $\seq{x}$ that is
reached from it via the branching direction~$y$;
node $\seq{x, y}$ is the \emph{parent} of node $\seq{x, y, z}$;
nodes $\seq{x, y}$ and $\seq{x, y, z}$ are \emph{descendants} of nodes
$\seq{}$ and~$\seq{x}$;
nodes $\seq{}$ and $\seq{x}$ are \emph{ancestors} of nodes
$\seq{x, y}$ and~$\seq{x, y, z}$;
and a node is a \emph{leaf} if it does not have any children.
All nodes in a well-ordered tree are well-ordered by the
\emph{lexicographic order} that is induced by the well-order on the 
branching directions;
for example, we have $\seq{x} < \seq{x, y}$, and
$\seq{x, y, z} < \seq{x, w}$ if $y < w$.
We define the \emph{depth} of a node to be the number of elements in
the eponymous sequence, the \emph{height} of a tree to be the maximum
depth of a node, and the \emph{size} of a tree to be the number of its
nodes. 

Parity progress measures assign labels to vertices of graphs with
vertex priorities, and the labels are nodes in a well-ordered tree. 
A \emph{tree labelling} of a graph with vertex priorities that do not
exceed a positive even integer~$d$ is a mapping from vertices of the
graph to nodes in a well-ordered tree of height at most~$d/2$.  
We write $\seq{m_{d-1}, m_{d-3}, \dots, m_\ell}$, for some odd~$\ell$, 
$1 \leq \ell < d$, to denote such nodes.
We say that such a node has an (odd) \emph{level}~$\ell$ and
an (even) level~$\ell-1$, and the root~$\seq{}$ has level~$d$. 
Moreover, for every priority~$p$, $0 \leq p \leq d$, we define the
\emph{$p$-truncation} 
$\seq{m_{d-1}, m_{d-3}, \dots, m_\ell}|_p$ in the following way:
$$
  \seq{m_{d-1}, m_{d-3}, \dots, m_\ell}|_p =
  \begin{cases}
    \seq{m_{d-1}, m_{d-3}, \dots, m_\ell} & \text{ for $p \leq \ell$},
    \\
    \seq{m_{d-1}, m_{d-3}, \dots, m_{p+1}} & \text{ for even $p > \ell$}, 
    \\
    \seq{m_{d-1}, m_{d-3}, \dots, m_p} & \text{ for odd $p > \ell$}.
  \end{cases}
$$
We then say that a tree labelling~$\mu$ of a graph $G=(V, E)$ with
vertex priorities $\pi : V \to \eset{0, 1, 2, \dots, d}$ is a
\emph{parity progress measure} if the following 
\emph{progress condition} holds for every edge $(v, u) \in E$:
\begin{compactenum}
\item
  if $\pi(v)$ is even then $\mu(v)|_{\pi(v)} \geq \mu(u)|_{\pi(v)}$;  
\item
  if $\pi(v)$ is odd then $\mu(v)|_{\pi(v)} > \mu(u)|_{\pi(v)}$. 
\end{compactenum}
It is well-known that satisfaction of such local conditions on every
edge in a graph is sufficient for every infinite path in the graph
satisfying the parity condition~\cite{Jur00,JL17}. 
Less obviously, it is also necessary, which can be, again, deduced
from the proof of positional determinacy of parity games due to
Emerson and Jutla~\cite{EJ91}.
In other words, parity progress measures are witnesses for the
property that all infinite paths in a graph satisfy the parity
condition. 
Like for the simpler co-B\"uchi condition, their appeal stems from the
property that they certify conditions that are global and infinitary
by verifying conditions that are local to every edge in the graph. 

Similar to the simpler co-B\"uchi progress measures, parity
progress measures may unfortunately require unbounded or even infinite
well-ordered trees to certify parity conditions on infinite graphs,
and hence we consider lazy parity progress measures, also inspired by
Klarlund's pioneering work~\cite{Kla91}. 
A~\emph{lazy tree} is a well-ordered tree with a distinguished subset
of its nodes called \emph{lazy nodes}.
For convenience, we assume that only leaves may be lazy and the root 
never is. 

A \emph{lazy parity progress measure} is a tree labelling $\mu$
of a graph $(V, E)$, where the labels are nodes in a lazy tree~$T$, 
such that:  
\begin{compactenum}
\item 
\label{lazyparity-one} 
  for every edge $(v, u) \in E$, 
  $\mu(v)|_{\pi(v)} \geq \mu(u)|_{\pi(v)}$; 
\item 
\label{lazyparity-two} 
  if $\pi(v)$ is odd then node $\mu(v)$ is lazy and its level is at
  least~$\pi(v)$; 
\item 
\label{lazyparity-three} 
  on every infinite path in $G$, there are infinitely many
  vertices~$v$, such that $\mu(v)$ is not lazy.  
\end{compactenum}

First we establish that existence of a lazy progress measure is
sufficient for all infinite paths in a graph to satisfy the parity
condition. 

\begin{lemma}
  If a graph has a lazy parity progress measure then all infinite
  paths in it satisfy the parity condition.  
\end{lemma}

\begin{proof}
  For the sake of contradiction, assume that there is an infinite path
  $v_1, v_2, v_3, \ldots$ in the graph for which the highest
  priority~$p$ that occurs infinitely often is odd.
  Let $i \geq 1$ be such that $\pi(v_j) \leq p$ for all
  $j \geq i$. 
  By condition~\ref{lazyparity-one}), we have:
  \begin{align}
    \label{align:descending}
    \mu(v_i)|_{p} \geq \mu(v_{i+1})|_{p} \geq \mu(v_{i+2})|_{p} \geq \ldots\,.
  \end{align}
  Let $i \leq i_1 < i_2 < i_3 < \cdots$ be such that
  $\pi(v_{i_k}) = p$ for all $k = 1, 2, 3, \dots$. 
  By condition~\ref{lazyparity-two}), all labels $\mu(v_{i_k})$, for
  $k = 1, 2, 3, \dots$, are lazy and their level in the tree is at
  least~$p$.
  By condition~\ref{lazyparity-three}), for infinitely many $k$, $\pi(v_{k})$ is not lazy, so infinitely many of the
  inequalities in~(\ref{align:descending}) must be strict, which
  contradicts the well-ordering of the tree~$T$. 
\end{proof}

Now we argue that existence of lazy parity progress measure is also
necessary for a graph to satisfy the parity condition. 
Moreover, we explicitly quantify the size of a lazy ordered tree the
labels from which are sufficient to give a lazy progress measure for a
layered dag, as a function of the width of the dag.
Before we do that, however, we introduce a simple operation that we
call a lazification of a finite ordered tree.
If $T$ is a finite ordered tree, then its \emph{lazification}
$\mathrm{lazi}(T)$ is a finite lazy tree that is obtained from~$T$ in
the following way: 
\begin{compactitem}
\item
  all nodes in $T$ are also nodes in $\mathrm{lazi}(T)$ and they are
  not lazy; 
\item
  for every non-leaf node $t$ in~$T$, $t$ has extra
  lazy children in the tree $\mathrm{lazi}(T)$, one smaller and one
  larger than all the other children, and one in-between every pair of
  consecutive children.  
\end{compactitem}
It is routine to argue that if a tree has $n$ leaves and it is of 
height at most~$h$ then its lazification $\mathrm{lazi}(T)$ has $O(nh)$
nodes and it is also of height~$h$. 

\begin{theorem}[Klarlund~\cite{Kla91}]
\label{theorem:lazy}
  If all infinite paths in a layered dag satisfy the parity condition
  and the width of the dag is at most~$n$, then there is a lazy parity
  progress measure whose labels are nodes in a tree that is a
  lazification of an ordered tree with at most~$n$ leaves.
\end{theorem}



\begin{proof}
%
%
  Klarlund's proof~\cite{Kla91} is very succinct and it heavily relies 
  on the result of Klarlund and Kozen on Rabin 
  measures~\cite{KK95}. 
  Our proof is not only self-contained but it also is more
  constructive and transparent.
  The hierarchical decomposition describes the fundamental structure
  of accepting run dags of alternating parity automata and it may be
  of independent interest.  
  The argument presented here is a generalization of the proof of
  Lemma~\ref{lemma:small-coBuchi-lazy-pm}---given in 
  Section~\ref{section:from-co-buchi-to-weak}---from co-B\"uchi 
  conditions to parity conditions.  

  Consider a layered dag
  $G = (V, E, \pi)$ where $\pi: V\rightarrow \eset{0, 1, 2, \dots, d}$. 
  For a priority~$p$, $0 \leq p < d$, we write $G^{\leq p}$ for the
  subgraph induced by the vertices whose priority is at most~$p$.

  We describe the following decomposition of~$G$. 
  Let $D$ be the set that consists of all vertices of the top even
  priority~$d$ in $G$, and $R_0$ all those vertices in the subgraph
  $G^{\leq d-1}$ that have finitely many descendants. 
  We say that those vertices are \emph{$(d-1)$-transient} in $G^{\leq d-1}$.
  In other words, $D \cup R_0$ is the set of vertices from which every path
  reaches (possibly immediately) a vertex of priority~$d$.

  Let $G_1 = G \setminus (D \cup R_0)$ be the layered dag obtained from~$G$ by
  removing all vertices in~$D \cup R_0$.
  Observe that every vertex in~$G_1$ has at least one successor and
  hence---unless $G_1$ is empty---all maximal paths are infinite.
  W.l.o.g., assume henceforth that $G_1$ is not empty.
  We argue that there must be a vertex in~$G_1$ whose all descendants
  have priorities at most~$d-2$;
  call such vertices \emph{$(d-1)$-safe} in~$G_1$.
  Indeed, otherwise it would be easy to construct an infinite path
  with infinitely many occurrences of the odd priority~$d-1$ and no 
  occurrences of the top even priority~$d$.

  Let $S_1$ be the set of all the $(d-1)$-safe vertices in~$G_1$.
  Let $H_1$ be the subgraph of $G$ induced by $S_1$, let $n_1$ be 
  the width of~$H_1$, and note that $n_1 > 0$. 
  Set $G'_1 = G_1 \setminus S_1$ to be the layered dag obtained 
  from~$G_1$ by removing all $(d-1)$-safe vertices in~$G_1$.

  Let $R_1$ be the set of all vertices in~$G'_1$ that have only
  finitely many descendants; call such vertices
  $(d-1)$-transient in~$G'_1$.
  Finally, let $G_2$ be the layered dag obtained from $G'_1$ by
  removing all the $(d-1)$-transient vertices in~$G'_1$.
  Note that the width of~$G_2$ is smaller than the width of~$G_1$ by
  at least~$n_1 > 0$.
  
  Unless graph $G_2$ is empty, we can now apply the same steps
  to~$G_2$ that we have described for~$G_1$, and obtain: 
  \begin{compactitem}
  \item
    the set $S_2$ of $(d-1)$-safe vertices in~$G_2$;
  \item
    the subgraph $H_2$ of~$G_2$ induced by $S_2$, which is of width 
    $n_2 > 0$; 
  \item
    the layered dag $G'_2$, obtained from $G_2$ by removing all the
    vertices in~$S_2$;
  \item
    the set $R_2$ of $(d-1)$-transient vertices in~$G'_2$;
  \item
    the layered graph~$G_3$, obtained from~$G_2$ by removing all
    vertices in~$S_2 \cup R_2$, and whose width is smaller than the
    width of~$G_2$ by at least~$n_2 > 0$.
  \end{compactitem}
  We can continue in this fashion, obtaining graphs
  $G_1, G_2, \dots, G_{k+1}$, until the graph~$G_{k+1}$, for some
  $k \geq 0$, is empty.
  Since the width of~$G$ is at most~$n$ and the widths of the graphs
  $G_1, G_2, \dots, G_k$ are positive (unless $k=0$), we have that 
  $k \leq n$ and $\sum_{i=1}^k n_i \leq n$. 

  The proces described above yields a hierarchical decomposition of
  the layered dag;
  we now define---by induction on~$d$---the tree that describes the
  shape of this decomposition.
  We then argue that the lazification of this tree provides the set of
  labels in a lazy parity progress measure. 


  In the base case $d=0$, the shape of the decomposition is 
  the well-ordered tree~$T$ of height $h = 0/2 = 0$ with only a root
  node~$\seq{}$.  
  It is straightforward to see that the function that maps every
  vertex onto the root is a (lazy) progress measure.

  For $d\geq 2$, note that all vertices in dags $H_1, H_2, \dots, H_k$
  have priorities at most~$d-2$.
  By the inductive hypothesis, there are trees $T_1, T_2, \dots, T_k$,
  of heights at most $h-1 = (d-2)/2$ and with at most $n_1, n_2,
  \dots, n_k$ leaves, respectively, which are the shapes of the
  hierarchical decompositions of dags $H_1, H_2, \dots, H_k$,
  respectively.

  We now construct the finite ordered tree~$T$ of height at most 
  $h = d/2$ that is the shape of the hierarchical decomposition
  of~$G$:
  let~$T$ consist of the root node~$\seq{}$ that has $k$ children, 
  which are the roots of the subtrees $T_1, T_2, \dots, T_k$, in that
  order.
  Note that the number of leaves of~$T$ is at most
  $\sum_{i=1}^k n_i \leq n$. 
  Consider the following mapping from vertices in the graph onto nodes
  in the lazification~$\mathrm{lazi}(T)$ of tree~$T$:
  \begin{compactitem}
  \item
    vertices in set~$D$ are mapped onto the root of $\mathrm{lazi}(T)$; 
  \item
    vertices in transient sets $R_0, R_1, R_2, \dots, R_k$ are mapped
    onto the lazy children of the root of $\mathrm{lazi}(T)$:
    those in $R_0$ onto the smallest lazy child, those in $R_1$ onto
    the lazy child between the roots of~$T_1$ and~$T_2$, etc.;
  \item
    vertices in subgraphs $H_1, H_2, \dots, H_k$ are inductively
    mapped onto the appropriate nodes in the lazy subtrees
    of~$\mathrm{lazi}(T)$ that are rooted in the $k$ non-lazy children
    of the root.
  \end{compactitem}
It is easy to verify that this mapping satisfies
conditions~\ref{lazyparity-one}) and~\ref{lazyparity-two}) of the
definition of a lazy parity progress measure.  
Condition~\ref{lazyparity-three}) is ensured by the fact that the root
of $T$ is not lazy and by the inductive hypothesis.
Recall that every infinite path satisfies the parity condition, thus
the highest priority $p$ seen infinitely often on a given path is
even.  
If $p=d$, the path visits infinitely often vertices labelled by the
root of~$T$. 
Otherwise, eventually the path contains only vertices in one of the
sets $S_i$ and we can use the inductive hypothesis.    
\end{proof}

\section{From parity to B\"uchi via universal trees}
\label{section:from-parity-to-buchi}

In this section we complete the proof of the main technical result of
the paper, which is a quasi-polynomial translation from alternating
parity automata to alternating weak automata.
The main technical tools that we use to design our
translation are lazy progress measures and universal
trees~\cite{JL17,CDFJLP19}, and the state space blow-up of the
translation is merely quadratic in the size of the smallest universal
tree.  
Nearly tight quasi-polynomial upper and lower bounds have recently
been given for the size of the smallest universal
trees~\cite{JL17,CDFJLP19} and, in particular, they imply that if the
number of priorities in a family of alternating parity automata is at
most logarithmic in the number of states, then the state space blow-up
of our translation is only polynomial.

\begin{theorem}
\label{theorem:main}
  There is a translation that given an alternating parity automaton 
  with $n$ states and $d$ priorities yields an equivalent alternating
  weak automaton whose number of states is polynomial if
  $d = O(\log n)$ and it is $n^{O\left(\lg (d/\lg n)\right)}$ if
  $d = \omega(\log n)$.
\end{theorem}

Before we proceed to prove the theorem, we recall the notion of
universal ordered trees. 
An \emph{$(n, h)$-universal (ordered) tree}~\cite{CDFJLP19} is an 
ordered tree, such that every finite ordered tree of height at 
most $h$ and with at most $n$ leaves can be isomorphically embedded
into it. 
In such an embedding, the root of the tree must be mapped onto the
root of the universal tree, and the children of every node must be
mapped---injectively and in an order-preserving way---onto the
children of its image.
In order to upper-bound the size of the blow-up in our parity to weak
translation, we rely on the following upper bound on the size of the 
smallest universal trees.

\begin{theorem}[Jurdzi\'nski and Lazi\'c~\cite{JL17}] 
\label{theorem:universal}
  For all positive integers $n$ and $h$, there is an
  $(n, h)$-universal tree with at most quasi-polynomial number of
  leaves.
  More specifically, the number of leaves is polynomial in~$n$ if
  $h = O(\log n)$, and it is $n^{\lg(h/\lg n)+O(1)}$ if
  $h = \omega(\log n)$. 
\end{theorem}
We also note that Czerwi\'nski et al.~\cite{CDFJLP19} have
subsequently proved a nearly-matching quasi-polynomial lower bound,
hence establishing that the smallest universal trees have
quasi-polynomial size.


In order to prove Theorem~\ref{theorem:main}, we establish the
following lemma that provides a translation from alternating parity
automata to alternating B\"uchi automata whose state-space blow-up is
tightly linked to the size of universal trees. 

\begin{lemma}
\label{lemma:parity-to-buchi}
  There is a translation that given an alternating parity automaton
  with $n$ states and $d$ priorities yields an equivalent alternating
  B\"uchi automaton whose number of states is $O(ndL_U)$ where $L_U$ 
  is the number of leaves in an $(n, d/2)$-universal ordered tree~$U$.  
\end{lemma}

Note that Theorem~\ref{theorem:main} follows from
Lemma~\ref{lemma:parity-to-buchi} by
Theorem~\ref{theorem:universal} and
Corollary~\ref{thm:Buchi-quadratic}. 

\begin{proof}(of Lemma~\ref{lemma:parity-to-buchi}) 
Given an alternating parity
automaton~$\A = (Q, q_0, \Sigma, \delta, \pi : Q \to \eset{0, 1,\ldots, d})$
with $n$ states, we now design an alternating B\"uchi automaton
that guesses and certifies a dag run of~$\A$ together with a
lazy parity progress measure on it. 
As for the co-B\"uchi to weak case, we first construct a \emph{safety}
automaton that simulates the automaton~$\A$ while guessing a lazy  
parity progress measure and verifying conditions~\ref{lazyparity-one})
and~\ref{lazyparity-two}) of its definition. 
Condition~\ref{lazyparity-three}) will be later handled by
turning the safety automaton into a B\"uchi automaton by 
appropriately assigning priorities $1$ or $2$ to all states in the
safety automaton.

Below we give a general construction of an alternating B\"uchi
automaton $\Bb_T$ from any lazy well-ordered tree $T$, and then we
argue that the alternating parity automaton $\A$ is equivalent to the
alternating B\"uchi automaton $\Bb_{\mathrm{lazi}(U)}$, for every 
$(n, d/2)$-universal tree~$U$.

Let $T$ be a lazy tree of width $n$ and height~$d/2$. 
The construction is by induction on~$d$. 
The safety automaton $\Ss_T$ has the following set of states, which
are pairs of an element of $Q$ and of a node in~$T$.
\begin{itemize}
\item
  If $d=0$, then the set of states of $\Ss_T$ is 
  $\left(Q \times \{ \seq{} \}\right) \cup \eset{\mathtt{reject}}$.
\item
  Otherwise, let $\seq{x_1}, \seq{x_2}, \ldots, \seq{x_k}$ be the
  children of the root, and $1\leq i_1 < i_2 < \ldots < i_m \leq k$ 
  are the indices of its leaves that are lazy.

  For $i \notin \eset{i_1,i_2,\ldots, i_m}$, let $T_i$ be the lazy
  subtrees of $T$ of height at most $d/2 - 1$ rooted in $\seq{x_i}$.  
  By induction, for all $i$, we obtain an alternating B\"uchi
  automaton that is obtained from the lazy tree $T_i$ and from the 
  alternating parity automaton $\A$ restricted to the states of
  priority up to $d-2$.
  Let $\Omega_i$ denote its set of non-$\mathtt{reject}$ states.
  They are pairs consisting of an element of $Q$ and of a node in a
  tree of height $d/2-1$;
  the latter is a sequence $\seq{m_{d-3}, m_{d-5}, \dots, m_\ell}$ of
  at most $d/2 - 1$ branching directions. 
  Let $\Gamma_i$ be the set consisting of the pairs
  $\left(q, \seq{x_i, m_{d-3}, m_{d-5}, \dots, m_\ell}\right)$ for
  $\left(q, \seq{m_{d-3}, m_{d-5}, \dots, m_\ell}\right) \in \Omega_i$.
  Set $Q^{(d)}$ (resp. $Q^{(<d)}$) the set of states of priority $d$ 
  (resp.~$<d$) in~$\A$.
  The states of $\Ss_T$ are defined as:    
$$
    \left(Q^{(d)}\times \{\seq{}\}\right) \: \cup \: 
    \left( Q^{(<d)} \times \{\seq{x_{i_1}}, \seq{x_{i_2}}, \ldots,
      \seq{x_{i_m}}\}\right) 
    \cup \: \bigcup_{i=1}^k \Gamma_i  \: \cup \: \eset{\mathtt{reject}}\,.
$$
\end{itemize}

The initial state is $(q_0, t)$ where $t$ is the largest tuple such
that $(q_0, t)$ is a state.
Let us now define the transition function:  
for every state $(q, t)$, and for every $a \in \Sigma$, the formula
$\delta'\big((q, t), a\big)$ is obtained from $\delta(q, a)$ by
replacing every occurrence of state~$q' \in Q$ by the disjunction 
(i.e., a non-deterministic choice)  
$$
   \bigvee \lreset{(q', t') \: : \: t|_{\pi(q)} \geq t'|_{\pi(q)}}\,,
$$
where every occurrence $(q', t')$ which is not in the set of states 
stands for the state $\mathtt{reject}$.

In other words, the safety automaton~$\Ss_T$ can be thought of as 
consisting of copies of $\A$, for each node of the tree $T$, in whose 
acceptance games Elvis always has the choice to stay in the current
copy of~$\A$ or to move to one of a smaller node with respect to the
priority of the current state. 
Since the transitions of the safety automaton~$\Ss_T$ always respect the
transitions of the original parity automaton~$\A$, an accepting
run dag of~$\Ss_T$ yields a run dag of~$\A$
(obtained from the first components of the states $(q, t)$)
and a labelling of its vertices by nodes in $T$
(obtained from the second components of the states $(q, t)$). 
It is routine to verify that the design of the state set and of the
transition function of the safety automaton~$\Ss_T$ guarantees that the
latter labelling satisfies condition~\ref{lazyparity-one})
and~\ref{lazyparity-two}) of the definition of a lazy parity 
progress measure.

In order to obtain the B\"uchi automaton $\Bb_T$ from the safety
automaton $\Ss_T$, it suffices to appropriately assign priorities~$1$
and~$2$ to all states:
we let the state $\mathtt{reject}$ and all states $(q, t)$ such that
$t$ is a lazy node in tree~$T$ have priority~$1$, and we let all
states $(q, t)$ such that $t$ is not a lazy node in tree~$T$ have
priority~$2$.
Note that this ensures that a run of~$\Bb_T$ is accepting if and only
if the tree labelling of a run dag of~$\A$ that the underlying safety
automaton~$\Ss_T$ guesses---in the second component of its
states---satisfies condition~\ref{lazyparity-three}) of the definition
of a lazy parity progress measure.

We now argue that if $U$ is an $(n, d/2)$-universal tree then the
alternating B\"uchi automaton $\Bb_{\mathrm{lazi}(U)}$ is equivalent
to the alternating parity automaton~$\A$. 
Firstly, all words accepted by~$\Bb_T$ for any finite lazy ordered
tree~$T$ are also accepted by~$\A$.
This is because---as we have argued above---every accepting run dag of
any such automaton~$\Bb_T$ yields both a run dag of~$\A$
(in the first state components) and a lazy parity progress measure on
it (in the second state components), and the latter certifies that the
former is accepting.    

It remains to argue that every word accepted by~$\A$ is also accepted
by~$\Bb_{\mathrm{lazi}(U)}$.
By Theorem~\ref{theorem:lazy}, for every accepting run dag of~$\A$,
there is a lazy progress measure whose labels are nodes in a tree
$\mathrm{lazi}(T)$, where $T$ is an ordered tree of height at
most~$d/2$ and with at most~$n$ leaves.
It is routine to verify that if an ordered tree can be isomorphically
embedded in another, then the same holds for their lazifications.
By $(n, d/2)$-universality of~$U$, it follows that $\mathrm{lazi}(T)$
can be isomorphically embedded in $\mathrm{lazi}(U)$. 
Therefore, for every word on which there is an accepting run dag
of~$\A$, the automaton $\Bb_{\mathrm{lazi}(U)}$ has the capacity to
guess the run dag of~$\A$ and to guess and certify a lazy progress
measure on it.
%
%
%

In order to conclude the $O(ndL_U)$ upper bound on the number of 
states of $\Bb_{\mathrm{lazi}(U)}$, it suffices to observe that if the
number of leaves in an ordered tree~$T$ of height~$h$ is~$L$ then the
number of nodes in $\mathrm{lazi}(T)$ is $O(hL)$. 
\end{proof}

\section{Open questions}

Our use of universal trees to turn alternating parity automata into
B\"uchi automata, like Boker and Lehtinen's~\cite{BL18} register
technique, does not exploit alternations
(although the further B\"uchi to weak translation does):
all transitions that are not copied from the original automaton are
non-deterministic. 
Can universal and non-deterministic choices be combined to
further improve these translations?
Can the long-standing $\Omega(n\log n)$ lower bound~\cite{Mic88} 
be improved, for example by a combination of the full-automata
technique of Yan~\cite{Yan08} and the recent lower bound techniques
for non-deterministic safety separating automata based on universal 
trees~\cite{CDFJLP19} and universal graphs~\cite{CF19}? 
  
\bibliography{concur2019-14}

\end{document}